\documentclass[a4]{article}
\usepackage{amssymb}
\usepackage{amsfonts}
\usepackage{amsmath}
\usepackage{latexsym}
\usepackage{a4wide} 

\usepackage{mathtools}
\usepackage{tikz, float}
\usepackage[british]{babel}

\usepackage[all]{xy}
 
 \def\set#1{{\{ #1\}}}
 \def\c#1{{\mathcal #1}}
 
\def\ws{{winning strategy}}

\def\O{\operatorname{Forb}}

\newcommand{\mc}{\mathcal}

\newcommand{\comp}{\mathbin{;}}
\newcommand\dom{\operatorname{dom}}
\newcommand\ran{\operatorname{ran}}
\newcommand{\Rel}{{\sf Rel}}  %Keep R for R(signature)
\newcommand{\Ltrel}{{\sf Lt} } %%cahnged R to L
\newcommand{\Trel}{{\sf T} }
\def\join{\cup}

\newcommand{\bea}{\begin{eqnarray*}}
\newcommand{\eea}{\end{eqnarray*}}

\newcommand{\ben}{\begin{enumerate}}
\newcommand{\een}{\end{enumerate}}

\newcommand{\bi}{\begin{itemize}}
\newcommand{\ei}{\end{itemize}}

\newtheorem{thm}{THEOREM}[section]
\newtheorem{lemma}[thm]{LEMMA}
\newtheorem{cor}[thm]{COROLLARY}
\newtheorem{pro}[thm]{PROPOSITION}

\newenvironment{proof}{
PROOF:
\begin{quotation}}{
$\Box$ \end{quotation}}

\def\Los{\L{o}\'{s}}
\def\Zar  {{Zarecki\u{\i} }}

\def\restr #1{{\restriction_{#1}}}

\newcommand{\ccup}{\mathbin{\sqcup\kern -6pt\sqcup}}
 
\title{The algebra of non-deterministic programs: demonic operations, orders and axioms}
\author{Robin Hirsch, Szabolcs Mikul\'as and Tim Stokes}

\begin{document}
\maketitle
\begin{abstract}
Demonic composition, demonic refinement and demonic union are alternatives to the usual ``angelic" composition, angelic refinement (inclusion) and angelic (usual) union defined on binary relations.  We first motivate both the angelic and demonic via an analysis of the behaviour of non-deterministic programs, with the angelic associated with partial correctness and demonic with total correctness, both cases emerging from a richer algebraic model of non-deterministic programs incorporating both aspects.  
\Zar\ has shown that the isomorphism class of algebras of binary relations 
under angelic composition and inclusion is finitely axiomatised as the class of ordered semigroups.
The proof can be used to establish that the same axiomatisation applies to binary relations under demonic composition and refinement, and a further modification of the proof can be used to incorporate a zero element representing the empty relation in the angelic case and the full relation in the demonic case. For the signature of angelic composition and union, it is known that no finite axiomatisation exists, and we show the analogous result for demonic composition and demonic union by showing that the same axiomatisation holds for both.  We show that the isomorphism class of algebras of binary relations 
with the ``mixed" signature of demonic composition and angelic inclusion has no finite axiomatisation.  As a contrast, we show that the isomorphism class of partial algebras of binary relations with the partial operation of constellation product and inclusion (also a ``mixed" signature) is finitely axiomatisable.
%; Zarecki\u{\i}'s original result is the special case of this positive result in which the partial product happens to be total because all the relations are left-total.
\end{abstract}

Throughout, let $X$ be a non-empty set.  Let $\Rel(X)$ denote the set of all binary relations on $X$, with $\Ltrel(X)$ and $\Trel(X)$ the sets of {\em left total} and {\em total} binary relations on $X$ respectively.  

\section{Introduction}

The \emph{Cayley representation} of a group $(G,\cdot, {}^{-1})$ represents an element $g\in G$ 
as the permutation $h\mapsto hg$ and shows that every group is isomorphic to 
a group of permutations under composition and inverse of permutations.  
The same representation works for monoids and represents each element as a total function.  
For a semigroup $(A, \circ)$, in order to ensure injectivity of the representation, 
we can use the base set $A\cup\set{e}$ (where $e$ is an adjoined identity element) 
and represent each element $a\in A$ as the total function $\set{(x, x\circ a): x\in A}\cup\set{(e, a)}$.

\Zar \cite{zarecki} showed how to extend this type of representation 
to ordered semigroups, showing that each is isomorphic to a semigroup of binary relations under composition and inclusion.  In other words, he showed that the set of axioms of ordered semigroups 
(associativity, partial order,  left and right monotonicity) 
exactly defines the class of structures isomorphic to 
concrete structures of binary relations on some base set $X$ under inclusion~$\subseteq$ and 
\emph{relational composition}~$\comp$ (sometimes called \emph{angelic composition}):
\begin{equation}
s\comp t =\{(x,y)\in X\times X: \exists z((x,z)\in s\land (z,y)\in t)\}.
\end{equation}
 Since the set of all binary relations on a non-empty set is an ordered semigroup under composition and inclusion, it immediately follows that the isomorphism class of semigroups of binary relations under composition and equipped with inclusion as a partial order is the class of ordered semigroups.  

One can also consider signatures containing the operation of union $\join$, which is join with respect to inclusion $\subseteq$.  
Any semigroup of binary relations under composition also closed under union is a semiring in which addition $\join$ is a semilattice operation, and in any such abstract semiring $(S,\cdot,+)$, $(S,\cdot,\leq)$ is an ordered semigroup, where $\leq$ is the induced partial order given by $s\leq t$ if and only if $t=s+t$.  However, in \cite{andreka}, it is shown that semirings of binary relations under composition and union have no finite axiomatisation.

There is interest in obtaining similar results when standard ``angelic" operations and orders are replaced by their ``demonic" variants.  These can be motivated from the theory of non-deterministic programs, and we do this in some detail shortly.  
\emph{Demonic composition}~$*$ is defined by
\begin{equation}\label{def:demonic}
s*t=\{(x,y):\exists z((x, z)\in s\wedge (z, y)\in t) \wedge
\forall w ((x, w)\in s\rightarrow\exists v (w, v)\in t)\}.
\end{equation}
where $s, t\in \Rel(X)$ with $X$ a non-empty set.  Equivalently, 
\begin{equation}\label{eq:fwd}
s*t=(s;t)\cap \set{(x, y): s(x)\subseteq dom(t)}  \end{equation}
where $dom(t)$ denotes the domain of the relation $t$ and $s(x)$ denotes the forward image $\set{y\in X:(x, y)\in s}$ of $x\in X$ under $s$.  

%Intuition may be gained by considering non-deterministic programs.  A single iteration of program $p$ may be represented as a binary relation over program states.    For programs $p, q$ write $pq$ for the program that first runs $p$ and if it terminates then runs $q$.    Thinking of total correctness, let $[p]$ be the set of pairs of states $(x, y)$ such that there is a finite sequence of $p$-transitions from $x$ to $y$ and there are no infinite sequences of $p$-transitions starting from $x$.    For arbitrary programs $p, q$ we have $[p]*[q]\subseteq[pq]$,\nb{wrong} although this inclusion can be proper.  Write $q'$ for the program: either do $q$ or make a single transition to a halt state.  Then $[p]*[q'] =[p(q')]$.  It has been argued that demonic composition~$*$ together with  demonic refinement order $\sqsubseteq$ (below) are important when one is interested in program refinement in the setting of total correctness of non-deterministic programs.  See~ for example.

Often associated with demonic composition is the so-called \emph{demonic refinement order}: for binary relations $s,t$,
\begin{equation}\label{eq:refine}
s\sqsubseteq t \iff ( \dom(t)\subseteq \dom(s)\;\wedge\;   s\restr{dom(t)} \subseteq t)
\end{equation} 
For links with computer science see \cite{DAD}, as well as \cite{vonW1} and~\cite{vonW2}, among many others, but see also the next section. 
Demonic refinement is a partial order on binary relations, and curiously both left and right monotonicity are valid for demonic composition over demonic refinement (for angelic composition over demonic refinement, neither monotonicity law holds).  Hence, the set of all binary relations on a  given set is 
an ordered semigroup under demonic composition and refinement; the key properties are noted in \cite[Section 3.1]{DAD}.  

Finally, there is a demonic analog of union, namely \emph{demonic join} $\ccup$, defined by
\begin{equation}
s\ccup t=\{(x,y) : (x,y)\in s\join t \wedge x\in dom(s)\cap dom(t)\}.
\end{equation}
Again, it is noted in \cite{DAD} that $(\Rel(X),*,\ccup)$ is a semiring with $\ccup$ a semilattice operation, and indeed $\ccup$ is join with respect to $\sqsubseteq$ as also noted there. 

A second alternative notion of composition we consider here
can be viewed as an extreme version of demonic composition, the partial operation we call \emph{constellation product}~$\cdot$.
For relations $s$ and $t$, $s\cdot t$ is defined if and only if the range of $s$ is included in the domain of $t$, and then it agrees with the usual relational (and indeed demonic) composition:
\begin{equation}\label{eq:constel}
s\cdot t=\begin{cases}
s\comp t &\text{if $\ran(s)\subseteq \dom(t)$}\\
\text{undefined} &\text{otherwise}.
\end{cases}
\end{equation}
It can be viewed as an extreme form of demonic product, since in order that $(x,y)\in s\comp t$ also be in $s\cdot t$, it is not just enough that $s(x)\subseteq dom(t)$ (as for $s*t$) but the whole range of $s$ must be contained in $dom(t)$, with the product $s\cdot t$ undefined otherwise.  We call this partial operation constellation product because, for partial functions at least, it is the partial operation which motivates the definition of constellations as in~\cite{inductive}.  Constellations in general have a domain operation as well as the partial binary operation, and may be viewed as ``one-sided categories".  So-called ``inductive" constellations were used to obtain partial counterparts for left restriction semigroups in \cite{inductive}, and general constellations are considered in more detail in~\cite{constell}.

Finally, some natural nullary operations on $\Rel(X)$ are given by the empty relation $\emptyset$, the diagonal relation $1'$ and the full relation $\nabla$.  (We use $\nabla$ rather than $1$ since the latter is used in what follows to refer to a general identity element in a semigroup.)

Let $X_0=X\cup \{0\}$ where $0\not\in X$, and define
$$\Ltrel_0(X)=\{\rho\in \Ltrel(X_0)\mid \mbox{ for all } s\in X_0, (0,s)\in \rho\mbox{ if and only if }s=0\},$$
a subsemigroup under composition of $\Rel(X_0)$.% the elements of which we call {\em left total with zero} binary relations on $X$.
%We view $\Rel(X)$ as being a subset of $\Rel(X_0)$.

In Section \ref{motive}, we introduce the relational models of computer programs used throughout, namely elements of $\Ltrel_0(X)$.  Here $X$ denotes the set of all program states for a computation, and $0$ will denote a `fail' state, or strictly a non-halting loop state.  We show how all the usual angelic and demonic operations and orders arise in terms of notions of partial and total correctness of programs.  Some of this is familiar in the computer science literature (see \cite{back} for example), but we believe consolidation of this material has value in its own right as well as to the current work in terms of motivation.  The material highlights the fact that the angelic operations and orders arise in connection with partial correctness of programs, and the demonic versions with total correctness.

For the remainder of the paper following Section \ref{motive}, we are concerned with axiomatisations.  In Section \ref{orderpos}, we are initially interested in axiomatisations of ``natural" signatures combining only angelic or only demonic operations and orders, as introduced in Section \ref{motive}.  We also consider the relational structure used to model programs in full generality, $\Ltrel_0(X)$ equipped with its relevant operations and ordering, since the same methods can be applied to it.

We mentioned that the isomorphism class of algebras of binary relations 
under angelic composition and inclusion is finitely axiomatised as the class of ordered semigroups.
In Subsection \ref{orderonly}, we show that the same axiomatisation applies to binary relations under demonic composition and refinement, as well as to $\Ltrel_0(X)$ under its natural ordered semigroup operations.  In Subsection \ref{orderzero}, a modification of the proof is used to incorporate a zero element representing the empty relation in the angelic case, the full relation in the demonic and total cases, and the constant function mapping everything to $0$ in the case of $\Ltrel_0(X)$.   
In Subsection \ref{union}, we consider various forms of union: for the purely angelic signature of angelic composition and union,  no finite axiomatisation exists \cite{andreka}, and we show that the signature involving demonic composition and demonic union satisfies precisely the same laws.  

Finally, in Section \ref{mixed}, we consider some arguably less natural ``mixed" angelic/demonic signatures.  We first show in Subsection \ref{neg} that the isomorphism class of algebras of binary relations  
with the signature of demonic composition and angelic inclusion has no finite axiomatisation, a result we must prove from first principles rather than relying on a known similar result.  As a contrast, we show in Subsection \ref{pos} that the isomorphism class of partial algebras of binary relations with the partial operation of constellation product and inclusion (also a ``mixed" signature) is finitely axiomatisable; Zarecki\u{\i}'s original result may be viewed as the special case of this positive result in which all the relations are assumed to be left-total, and so the partial product is everywhere-defined.  A further corollary is a finite axiomatisation for just $\{\cdot\}$ alone, applicable to both relations and partial functions.

\section{Motivation and basic facts}  \label{motive}

We next describe how the angelic and demonic operations and relations, including the constellation product, arise in terms of non-deterministic computer programs.

\subsection{How to model non-deterministic, possibly non-halting programs}

We interpret programs as giving rise to actions on a state space $X$ which, given any $x\in X$, can move and halt in another state, in a non-deterministic fashion, though it may also have non-terminating runs starting from $x$.  The idea here is that for $x\in X$, the program will always do something, but it could fail to halt in some cases but halt in others and with possibly distinct results. This scenario cannot be directly modelled using binary relations on $X$ since there is no way to tell whether a program with associated relation $\rho$ might not halt when starting at state $x\in dom(\rho)$, that is, a state $x$ from which termination is at least possible.

We therefore view every program $s$ acting on $X$ as giving rise to some relation $\rho_s\in \Ltrel_0(X)$, as follows.
\bi
%\item for all $x\in X_0$, $(0,x)\in \rho_s$ if and only if $x=0$;\nb{already in def of Lt0}
\item for all $x\in X$, $(x,0)\in \rho_s$ if and only if it is possible that $s$ does not halt when starting at state $x$;
\item for all $x,y\in X$, $(x,y)\in \rho_s$ if and only if, when starting at state $x$, \/$s$ can halt in state $y$.
\ei
This is in essence the approach taken by Back in \cite{back}. Note that the mapping $s\mapsto \rho_s$ is not assumed to be either injective or surjective.  

For programs $s,t$ acting on the state space $X$, let $st$ denote ``$s$ followed by $t$".  We also introduce the notion of non-deterministic choice of programs: $s|t$.  Two other important programs are ``skip" and ``abort" (or more accurately, ``never halt").

\begin{pro}
For all programs $s,t$, $\rho_{st}=\rho_s \comp \rho_t$, $\rho_{s|t}=\rho_s\cup \rho_t$, $\rho_{abort}=\{(x,0)\mid x\in X_0\}$ and $\rho_{skip}=1'$, the diagonal relation on $X_0$.
\end{pro}

Two different programs can give rise to the same member of $\Ltrel_0(X)$, but from now on we blur the distinction between programs and actions, and view programs as elements of $\Ltrel_0(X)$, with the operations and orders of primary interest on it being $\comp$, $\cup$, $\subseteq$.  We also use the notation ${\bf 0}$ for $\rho_{abort}$ above.
% So for example if $s$ might not halt at $x\in X$ but $t$ always does, then $s|t$ might not halt at $x$, and so on.

\subsection{Some important mappings}

Denote by $\rho^a,\rho^d\in \Rel(X)$ the relations obtained from $\rho\in \Ltrel_0(X)$ as follows: 
\bi
\item $\rho^a=\rho\cap(X\times X),$ the {\em angelic restriction} of $\rho$, and
\item $\rho^d=\{(x,y)\in \rho\mid x\in X, (x,z)\in \rho\Rightarrow z\neq 0\}),$ the {\em demonic restriction} of $\rho$.
\ei

Clearly $\rho^d\subseteq \rho^a\subsetneq \rho$ for all $\rho\in  \Ltrel_0(X)$.  Also, every binary relation $R$ on $X$ is an angelic restriction of an element of $\Ltrel_0(X)$, and indeed  is also a demonic restriction of a (possibly different)  element of $\Ltrel_0(X)$, (e.g.$R$ is the angelic restriction of  $R\cup{\bf 0}$ and $R$ is both the angelic restriction and the demonic restriction of  $R\cup\set{(x, 0): x\in X_0\setminus dom(R)}$).  The angelic restriction $\rho^a_s$ is perhaps the usual way to view a program $s$ as a relation, but it omits important information present in $\rho_s$, namely those $x\in dom(\rho_s)$ for which $s$ might not halt.  Still, it accurately describes the states that can possibly be reached from a given one using $s$.  The demonic restriction of $\rho^d_s$ excludes from the domain of $\rho_s$ any $x\in X$ from which non-termination is possible. Note that $\rho_s$ can be recovered from knowledge of both $\rho^a_s$ and $\rho^d_s$.

\begin{pro}  \label{homom}
For $\rho,\tau\in \Ltrel_0(X)$,
\bi
\item $\rho=\tau$ if and only if $\rho^a=\tau^a$ and $\rho^d=\tau^d$,
\item $(\rho\comp \tau)^a=\rho^a\comp\tau^a$ and $(\rho\comp \tau)^d=\rho^d*\tau^d$,
\item $(\rho\cup \tau)^a=\rho^a\cup\tau^a$ and $(\rho\cup \tau)^d=\rho^d\ccup\tau^d$,
\item $1'^a=1'^d=1'_X$ (the diagonal relation on $X$) and ${\bf 0}^a={\bf 0}^d=\emptyset$.  
\ei
\end{pro}

It follows that $(\Rel(X),\comp)$ and $(\Rel(X),\comp,\cup)$, as well as $(\Rel(X),*)$ and $(\Rel(X),*,\ccup)$, are homomorphic images of $(\Ltrel_0(X),\comp)$ and $(\Ltrel_0(X),\comp,\cup)$ respectively (under $\rho\mapsto \rho^a$ or $\rho\mapsto \rho^d$ respectively), and this persists when we add one or both of the related nullary operations as in the fourth part above.  Moreover the first part implies that $\Ltrel_0(X)$ under any signature $S$ contained in $\set{\comp,\cup,1',{\bf 0}}$ is a subdirect product of copies of $\Rel(X)$ equipped with the corresponding angelic and demonic signatures respectively.

But there are other important mappings going in the other direction.  Let $X$ be a non-empty set and denote by $\nabla$ the full relation on $X$, and for each $\rho\in \Rel(X)$, define
$$\nabla_{\rho}=\{(x,y)\mid x\in X\backslash \dom(\rho), y\in X\}\in \Rel(X),$$ 
the full relation on $X$ restricted in domain to the complement in $X$ of the domain of $\rho$.  Then define the mappings 
\begin{align*}
\psi_1:& \Rel(X)\rightarrow \Ltrel_0(X),\ &\rho&\mapsto \rho\cup {\bf 0},\\
\psi_2:& \Rel(X)\rightarrow \Trel(X_0),\ &\rho&\mapsto \rho\cup \nabla_{\rho}\cup\{(0,y)\mid y\in X_0 \}   \\
\psi_3: &\Rel(X)\rightarrow \Ltrel_0(X),\ &\rho&\mapsto \rho\cup \nabla_{\rho}\cup \{(0,0)\}.
\end{align*}
For each $\rho\in \Rel(X)$, %$\rho\subseteq \psi_i(\rho)$ for all $i=1,2,3$, and 
$\psi_3(\rho)\subsetneq \psi_2(\rho)$.

\begin{pro}  \label{embeddings}
Each of the mappings $\psi_1,\psi_2,\psi_3$ is structure-preserving as follows.
\bi
\item The mapping $\psi_1$ is an embedding $(\Rel(X),\comp,\cup,\emptyset)\rightarrow (\Ltrel_0(X),\comp,\cup,{\bf 0})$.
\item The mapping $\psi_2$ is an embedding $(\Rel(X),*,\ccup,\emptyset)\rightarrow (\Trel(X_0),\comp,\cup,\nabla)$.
\item The mapping $\psi_3$ is an embedding $(\Rel(X),*,\ccup,1')\rightarrow (\Ltrel_0(X),\comp,\cup,1')$.
\ei
\end{pro}

In the third of these, the first reference to $1'$ is as the diagonal relation on $X$, and the second is as the diagonal relation on $X_0$, and similarly, in the second, $\nabla$ is the full relation on $X_0$.

\subsection{Program correctness, refinement and relational signatures}

We have the notion of a Hoare triple $(e)s(f)$ where $e,f$ are predicates on $X$ and $s$ is a program acting on $X$.  We say it is partially correct if any terminating  execution of $s$ from a state  satisfying $e$ leads to a state satisfying $f$.  We say it is totally correct if  all executions of $s$ from a state satisfying $e$ terminate in a state satisfying $f$.   So total correctness implies partial correctness.  We view each test as a member of $\Ltrel_0(X)$ that is a restriction $e$ of the diagonal to its truth set $Y$ in $X$, together with all $(s,0)$ where $s\in X_0\backslash Y$, and then $e^a=e^d$ is again this restriction of the diagonal to $Y$.  We write $\alpha\leq\beta$ when $\alpha,\beta$ are such restrictions of the diagonal on $X$ for which the domain of $\alpha$ is a subset of the domain of $\beta$ (that is, $\alpha\comp\beta=\alpha$).

For $\rho\subseteq X\times X$, writing $D(\rho)$ for the diagonal relation restricted to $dom(\rho)$, it is easy to see that 

\bi
\item $(e)\rho(f)$ is partially correct if and only if $e;\rho;f=e;\rho$, 
\item $(e)\rho(f)$ is partially correct if and only if $e^a;\rho^a;f^a=e^a;\rho^a$,
\item $(e)\rho(f)$ is totally correct if and only if $e^d*\rho^d*f^d=e^d*\rho^d$ and $e^d\leq D(\rho^d)$,
\item $(e)\rho(f)$ is totally correct if and only if $(e\cdot \rho^d)\cdot f$ exists (and hence equals $e\cdot \rho^d$).
\ei

In particular then by Proposition \ref{homom}, $(e)\rho\comp\tau(f)$ is totally correct if and only if $e*\rho^d*\tau^d*f=e*\rho^d*\tau^d$ and $e\leq D(\rho^d*\tau^d)$, emphasising the importance of demonic composition when considering total correctness of program concatenations.  There are other obvious corollaries of Proposition \ref{homom} involving the other relational operations.  It is also not hard to show that $(e)\rho\comp\tau(f)$ is totally correct if and only if $((e\cdot \rho^d)\cdot \tau^d)\cdot f$ exists, showing a connection between  constellation product of general relations and total correctness.  So all of the notions of composition of binary relations we have discussed are relevant when considering program correctness.  Next we consider the orders.
%Using the fact that $s*t=A(s;A(t));s;t$ for all $s,t\inm Rel_X$, it is then an easy exercise to show the following.

%\begin{pro}
%Suppose $s$ is a program with $e,f$ tests on $X$.
%\bi
%\item The triple $(e)s(f)$ is partially correct if and only if $e\leq A(\rho^a_s;A(f))$ in $Rel(X)$.
%\item The triple $(e)s(f)$ is totally correct if and only if $\rho_e\leq D(\rho^d_s*f)$ in $Rel(X)$.
%\item  
%\ei
%\end{pro}

We say $s$  \emph{partially refines} $t$ if for all tests $e,f$, whenever $(e)t(f)$ is  partially correct, so is $(e)s(f)$.  We say $s$ \emph{totally refines} $t$ if for all tests $e, f$, whenever $(e)t(f)$ is totally correct, so is $(e)s(f)$. % So refinement is ``correctness-preserving",  either partial  or total.  
This is the approach of Back in \cite{back}. 
%We say $s$ refines $t$ if for all tests $e,f$, whenever $(e)t(f)$ is correct, so is $(e)s(f)$.  So refinement is ``correctness-preserving".  This is the approach of Back in \cite{back}.  Evidently this will depend on whether we mean ``partially correct" or ``totally correct". 

\begin{pro}
Let $\rho,\tau\in \Ltrel_0(X)$ be (the actions of) programs acting on the same state space. 
\bi
\item  $\rho$ partially refines $\tau$ if and only if $\rho^a\subseteq \tau^a$.  So for partial correctness, refinement corresponds to relational inclusion.
\item  $\rho$ totally  refines $\tau$ if and only if $\rho^d\sqsubseteq \tau^d$.  So for total correctness, refinement corresponds to demonic refinement.
\item $\rho$  partially and totally refines $\tau$  if and only if $\rho\subseteq \tau$.
\ei
\end{pro}

We have the following easy consequence of Proposition \ref{embeddings}

\begin{cor}  \label{embeddings2}
Each of the mappings $\psi_1,\psi_2,\psi_3$ is structure-preserving as follows.
\bi
\item The mapping $\psi_1$ is an embedding $(\Rel(X),\comp,\subseteq,\emptyset)\rightarrow (\Ltrel_0(X),\comp,\subseteq,{\bf 0})$.
\item The mapping $\psi_2$ is an embedding $(\Rel(X),*,\sqsubseteq,\emptyset)\rightarrow (\Trel(X_0),\comp,\subseteq,\nabla)$.
\item The mapping $\psi_3$ is an embedding $(\Rel(X),*,\sqsubseteq,1')\rightarrow (\Ltrel_0(X),\comp,\subseteq,1')$.
\ei
\end{cor}

Note that the partial orders $\subseteq,\sqsubseteq$ induce quasiorders on $\Ltrel_0(X)$ in the obvious way: for programs $\rho,\tau$, we say $\rho\lesssim \tau$ if $\rho^a\subseteq \tau^a$, and $\rho\triangleleft \tau$ if $\rho^d\sqsubseteq \tau^d$.  These are $\sqsubseteq_1$ and $\sqsubseteq_2$ in \cite{back}, and $\subseteq$ in $\Ltrel_0(X)$ is their intersection and is a partial order on it, the so-called approximation order considered in \cite{back}.  Using these relations, results along the lines of the above proposition were given in \cite{back}.  Purely algebraic proofs using relation algebra are not hard, but we do not give them here.

\section{Results for ``unmixed" signatures}  \label{orderpos}

%Recall that ${\bf 0}=\{(x,0\mid x\in X_0\}\in \Ltrel_0(X)$.  Denote by $1$ the full relation on a set; although this does not have such an obvious motivation in terms of programs as the other operatinos we have so far considered, it is fairly natural from a relation algebra perspective and turns out to arise quite naturally in what follows.

For the remainder, let $S$ be a signature contained in $(0,  {\bf 0}, 1', \nabla,\subseteq, \sqsubseteq, \join, \ccup, \comp, *)$.   For $S\subseteq (1', \nabla, \subseteq, \sqsubseteq, \join, \ccup, \comp, *)$  let $R(S), L(S)$ and $T(S)$ denote the closure under isomorphism of the classes of all (respectively all left total,  all total) algebras of binary relations, closed under the operations in $S$ and equipped with the relations in $S$.  So for example  if $\join\in S$ then the set must be closed under union, and so on.    For $S\subseteq ({\bf 0}, 1', \subseteq, \sqsubseteq, \join, \ccup, \comp, *)$, let $L_0(S)$ be the closure under isomorphism of the class of all sets of binary relations contained in $\Ltrel_0(X)$ for some set $X$, closed under operators in $S$ and including $\set{(x, 0)|x\in X_0}$ when ${\bf 0}\in S$.  We also allow $0\in S$ when defining $R(S)$ and require that the algebra of binary relations includes the empty relation in this case.
%We allow ${\bf 0}\in S$ only  in the case of $L_0(S)$ and we allow $\nabla\in S$  only in the other cases, $R(S), L(S)$ and $T(S)$.  
The operators $\comp$ and $*$, are interchangeable in any signature $S$ for $L(S), T(S)$ or  $L_0(S)$, since the demonic composition of two left-total relations equals their angelic composition, similarly $\cup, \ccup$ are interchangeable over left-total relations.  If $\Sigma$ is a set of $S$-formulas and $\c K$ is a class of $S$-structures we say that $\Sigma$ \emph{defines} or \emph{axiomatises} $\c K$ if $\c K$ is the class of all models of $\Sigma$.

\subsection{Signatures axiomatised as ordered semigroups}  \label{orderonly}

As discussed above,% in \cite{zarecki}, \Zar showed the following. %%%%We mention Z by name 11 times at least, so cut down.

\begin{thm} \label{zar}
$R(\subseteq,\comp)=L(\subseteq,\comp)$, the class of ordered semigroups \cite{zarecki} .  
\end{thm}

The proof represents each element $a$ of an ordered semigroup $\mathcal{A}=(A, \leq, \cdot)$  as a binary relation over the base $A^{1'}=A\cup\set{1'}$, carrier set of the semigroup ${\mathcal A}^{1'}$ obtained from ${\mathcal A}$ by adjoining an identity ${1'}$.  Then $a\in A$ is represented as
\begin{equation}\label{eq:Z}
\rho_a=\set{(x, y)\mid x\in A^{1'}, \; y\in A,\; y\leq xa}.
\end{equation}
where the inequality on the right is evaluated in $\c A$.
Observe that the representation of each element $a$ is a left total binary relation on the set $A^{1'}$ (every element of $A^{1'}$ appears in the domain of the relation); then because $L(\subseteq,\comp)\subseteq R(\subseteq,\comp)$, a class of ordered semigroups, the result follows.  

%Note that for $(x,y)\in\rho_a$ as in (\ref{eq:Z}) above, $y\neq 1$, so $(x,y)\in A^1\times A$ always. 

 It suits our purposes better to consider a variant of the \Zar approach in which we assume the ordered semigroup is a monoid already, with identity ${1'}$.  The argument that $a\mapsto \rho_a$ is a faithful representation of ${\mc A}$ within $(\Rel(A),\comp,\subseteq)$, and indeed within $(\Ltrel(A),\comp,\subseteq)$, is then essentially identical to that given in \cite{zarecki} (but note that ${1'}$ is not represented as the diagonal relation in general). So we have the following.

\begin{lemma} \label{ordmonoid}
Let ${\mc A}$ be an ordered monoid.  Then  the ordered semigroup reduct of ${\mc A}$ embeds in $(\Ltrel(A^{1'}),\comp,\subseteq)$ via $a\mapsto \rho_a$ where
\begin{equation}\label{eq:Z1}
\rho_a=\{(x,y)\mid x,y\in A, \c A\models y\leq xa\}.
\end{equation}
\end{lemma}

We can easily recover \Zar's more general result  for ordered semigroups as follows.   A \emph{compatible unital}  extension of an ordered semigroup $\c A=(A, \leq, ;)$ is an ordered monoid $\c A^{1'}=(A^{1'}, \leq, {1'}, ;)$ where $A^{1'}=A\cup\{{1'}\},\; {1'}\not\in A$ and the $\set{\leq, ;}$-reduct of the restriction to $A$ is $\c A$.  A compatible unital  extension may be obtained by specifying whether $a\leq {1'}$ and whether ${1'}\leq a$ for all $a\in A$ and must satisfy $a\geq {1'} \rightarrow (a;b\geq b\wedge b;a\geq b)$ and $a\leq {1'}\rightarrow (a;b\leq b\wedge b;a\leq b)$, for all $a, b\in A$.  For example, a compatible unital extension may be obtained by letting ${1'}$ be unrelated to all elements but itself. Hence,

\begin{lemma}  \label{compembed}
Every ordered semigroup $\c A=(A, \leq, ;)$ can be represented as left-total binary relations under inclusion and composition.
\end{lemma}

\begin{proof}
Take a compatible unital extension $\c A^{1'}$ of $\c A$ and embed $\c A$ in  $(\Ltrel(A^{1'}), \subseteq, \comp)$ by
\begin{equation}\label{eq:ZZ}
\rho_a=\{(x,y)\mid x, y\in A^{1'}, \; \c A^{1'}\models y\leq xa\},
\end{equation}
for $a\in A$.
\end{proof}

Recall that any semigroup of binary relations under demonic composition  equipped with the refinement order is an ordered semigroup.  Conversely, \eqref{eq:Z} gives a representation of an ordered semigroup ${\mc A}$
as an ordered semigroup of left total binary relations under composition and inclusion; but for such binary relations, 
ordinary and demonic composition coincide, as do inclusion and demonic refinement, so this also shows ${\mc A}\in \Rel(\sqsubseteq,*)$. Hence we have the following.

\begin{thm}   \label{lem:OS} 
$R(\sqsubseteq, *)$ is the class of ordered semigroups. 
\end{thm}

Of course any subsemigroup of $(\Ltrel_0(X),\comp,\subseteq)$ is an ordered semigroup.  In fact, the converse also holds (up to isomorphism) and so we have the following.

\begin{thm}   \label{lem:OS} 
$L_0(\leq, \comp)$ is the class of ordered semigroups. 
\end{thm}
\begin{proof}
First adjoin a zero $0$ as a smallest element to the ordered semigroup ${\mc A}$, giving ${\mc A}^0$.  Then by Lemma \ref{compembed}, ${\mc A}^0$, and hence ${\mc A}$, embeds as an ordered semigroup in $(\Ltrel((A^0)^{1'}),\comp,\subseteq)$ via $s\mapsto \rho_s$.   But for any $s\in A$, $(0,y)\in \rho_s$ says $y\leq 0;s=0$ so $y=0$, and so $\rho_s\in \Ltrel_0(A^{1'})$, so ${\mc A}$ embeds in $\Ltrel_0((A^0)^{1'})$
\end{proof}  

\subsection{Introducing a zero element}  \label{orderzero}

Now we introduce a zero element to the signatures being considered, standing for various things depending on context. Perhaps the simplest extension of the signature of ordered semigroups  that distinguishes the angelic and demonic cases is ordered semigroups with zero representing the empty relation $\emptyset$ (note that  $\emptyset$ is not left total, so the empty set is not considered for $L_0(S)$).  Note that $\emptyset\subseteq s$ for all binary relations $s$, but on the other hand $s\sqsubseteq \emptyset$ for all $s$.  

In $\Ltrel_0(X)$, there is a  different zero element, namely ${\bf 0}=\{(x,0)\mid x\in X\cup \{0\}\}$.  Note that with respect to the inclusion order, ${\bf 0}$ is minimal in $\Ltrel_0(X)$.  
Let us say $(S,\cdot,\leq,0)$ is an {\em ordered semigroup with weak zero} providing $(S,\cdot,\leq)$ is an ordered semigroup having a semigroup zero $0$ such that $s\leq 0\Rightarrow s=0$.  Hence $(\Ltrel_0(X),\comp,\subseteq, {\bf 0})$ is an ordered semigroup with weak zero, whence so is any subalgebra.  Conversely, we have the following.

\begin{lemma}  \label{gen}
	Suppose ${\mc A}=(A,\cdot, 0,\leq)$ is an ordered semigroup with weak zero, and let ${\mc A^{1'}}$ be any compatible unital extension in which ${1'}\leq 0$ is false.  Then letting $X=A^{1'}\backslash \set0$, ${\mc A}$ embeds in $(\Ltrel_0(X),\comp,{\bf 0},\subseteq)$ via $s\mapsto \rho_s$,   see \eqref{eq:ZZ}.
\end{lemma}
\begin{proof}  
First note that such a compatible unital extension ${\mc A^{1'}}$ exists, as we noted above.  By Lemma \ref{compembed}, the mapping $s\mapsto \rho_s$ defined as above for each $s\in A$ is an embedding of $(A,\cdot,\leq)$ within $(\Ltrel(X\cup\set{1'}),\comp,\subseteq)$.  But for all $s\in A\cup\set0$,\/ $(0,x)\in \rho_s$ if and only if $x\leq 0s=0$, that is, $x=0$.  So $\rho_s\in \Ltrel_0(X\cup\set{1'})$.  

It remains to show that $0$ is represented as ${\bf 0}$.  But for $x\in A$,\/  $(x,y)\in\rho_0$ if and only if $y\leq x0=0$, or $y=0$.  So $\rho_0={\bf 0}$, as required.
\end{proof}

Hence

\begin{thm}\label{three}
$L_0(\comp,\subseteq,{\bf 0})$ is the class of ordered semigroups with weak zero.
\end{thm}

We say $(S,\cdot,0,\leq)$ is an {\em ordered semigroup with zero} if $(S,\cdot,\leq)$ is an ordered semigroup, $0$ is a zero in $S$, and $0\leq x$ for all $x$.  So every ordered semigroup with zero is an ordered semigroup with weak zero.  The members of $R(\comp,\subseteq,\emptyset)$ are ordered semigroups with zero.  

%An ordered semigroup can be turned into an ordered semigroup with zero or its dual by adjoining a zero element.  
%
%\begin{lemma} \label{add0}
%If $(S,\cdot,\leq)$ is an ordered semigroup, let $S^0$ be $S$ with zero adjoined in the usual way, and extend $\leq$ to $S^0$ by setting $0\leq s$ (resp. $s\leq 0$) for all $s\in S$.  Then $(S^0,\cdot,\leq)$ is an ordered semigroup with zero (resp. dual ordered semigroup with zero) into which $(S,\cdot,\leq)$ embeds.
%\end{lemma}

Zarecki\u{\i}'s approach in \cite{zarecki} represents every semigroup element as a left total relation, so it cannot possibly represent {\em any} element of an ordered semigroup as the empty relation.  However, we can use Lemma \ref{gen} to obtain what we want.
% in any ordered semigroup with zero $S$, $\rho_0$ is a particular left total relation, namely
%$$\rho_0=\{(x,y)\in S\times S\mid y\leq x0\}\cup \{(1,y)\mid y\leq 0\}=\{(x,0)\mid x\in S^1\},$$
%so is not the empty relation.  

\begin{thm}  \label{one}
$R(\comp,\subseteq,\emptyset)$ is the class of ordered semigroups with zero.
\end{thm}
\begin{proof}  
Suppose ${\mc A}$ is an ordered semigroup with zero.  Define $\c A^{1'}$ by adjoining an identity element ${1'}$ with $0<{1'}$ and ${1'}$ otherwise unrelated to other elements. This is a compatible unital extension  to which Lemma \ref{gen} applies.  So we may represent ${\mc A}$ within $\Ltrel_0(X)$ as in that result, with $X=A^{1'}\backslash\{0\}$, note that we delete $0$ from $\c A^{1'}$, but put it back in the base of relations in  $\Ltrel_0(X)$. Under this representation, we have, for all $s\in A$ and $x\in A^{1'}$, 
$(x,0)\in \{(x,y)\mid y\leq xs\}=\rho_s$ since $0\leq xs$ for all $x\in A^{1'}$, so $\rho_s=\psi_1((\rho_s)^a)$, and so $\rho_s\mapsto (\rho_s)^a$ is one-to-one and hence is an embedding of the representation of $S$ into $(\Rel(S^{1'}),\comp,\subseteq,\emptyset)$ by Proposition \ref{homom}. 
\end{proof}

%Now it is not hard to check that the set of left total binary relations $Rel_0(X)\subseteq Rel(X_0)$, where $X_0=X\cup\{0\}$, given by
%$$Rel_0(X)=\{(x,0)\mid x\in X\}\cup \{(0,0)\}$$
%is an ordered subsemigroup of $(Rel(X_0),;,\subseteq)$ isomorphic to $(Rel(X),;,\subseteq)$, via the mapping that takes $s\in Rel_0(X)$ to $s'\in Rel(X)$ by simply omitting all ordered pairs from $s$ that involve $0$.  (The corresponding fact relating partial functions to $0$-preserving transformations is perhaps more familiar.)  Applying this to our representation using Zarecki\u{\i}'s method of the ordered semigroup with $0$ gives a representation of $S$ within $S^1\backslash \{0\}$, $s\mapsto \tau_s$, such that $\tau_0=\emptyset$ as required.  We therefore have established that our desired axiomatisation is simply ordered semigroups with zero satisfying the law $0\leq x$.
%
%In case the reader finds the above argument slightly unsatisfactory, we now give a direct proof in full. 

%However, returning to the case in which $0$ is smallest, we can use a straightforward modification of Zarecki\u{\i}'s proof to yield the desired result.

As noted earlier, with respect to $\sqsubseteq$, $\emptyset$ is the largest member of $\Rel(X)$.
We say $(S,\cdot,0,\leq)$ is a {\em dual ordered semigroup with zero} if its dual is an ordered semigroup with zero, so that it satisfies $x\leq 0$ for all $x$.   So $(\Rel(X),*,\sqsubseteq,\emptyset)$ is a dual ordered semigroup with zero.  For a different but related example, note that in the subsemigroup $(\Trel(X),\comp)$ of $(\Rel(X),\comp)$, the full relation $\nabla$ is a zero element that is largest with respect to inclusion, and so $(\Trel(X),\comp,\nabla,\subseteq)$ is also a dual ordered semigroup with zero, whence so is any subalgebra. Indeed we obtain the following.

\begin{thm}  \label{cort}
$T(\comp,\subseteq,\nabla)$ is the class of dual ordered semigroups with zero.
%, and $L(\comp,\subseteq,\nabla)$ is the class of ordered semigroups with left zero that is the largest member.
\end{thm}
\begin{proof}  
Let ${\mc A}=(A,\cdot,0,\leq)$ be a dual ordered semigroup with zero.  Adjoin an identity element satisfying only ${1'}<0$, easily seen to give a compatible unital extension ${\mc B}$ of ${\mc A}$.  By Lemma \ref{compembed}, ${\mc A}$ embeds in $(\Rel(A^{1'}),\comp,\subseteq)$ via $s\mapsto \rho_s$,  \eqref{eq:ZZ}.  Now under this representation, for all $s\in A$, $y\leq 0s=0$ for all $y\in A^{1'}$, so $\rho_s$ is total, and $\rho_0=\{(x,y)\mid x,y\in A^{1'},y\leq 0\}=\nabla$ on $A^{1'}$.  So ${\mc A}$ embeds in $(\Trel(X),\comp,\subseteq,\nabla)$.  The converse was just noted.  
\end{proof}
  
As noted already, $(\Rel(X),*,\emptyset,\sqsubseteq)$ is a dual ordered semigroup with zero, and the hoped-for converse does hold, giving a demonic counterpart of Theorem \ref{one}.

\begin{thm}  \label{two}
$R(*,\sqsubseteq,\emptyset)$ is the class of dual ordered semigroups with zero.
\end{thm}
\begin{proof} 
Represent the dual ordered semigroup with zero ${\mc A}=(A,\cdot,0,\leq)$ precisely as in the previous proof, as a subalgebra ${\mc A}'$ of $(\Trel(X),\comp,\subseteq,\nabla)$.  Next, note that for each $s\in A$, $\rho_s\in ran(\psi_2)$, so ${\mc A}'\subseteq ran(\psi_2)$, and so by Proposition \ref{embeddings2}, $\psi_2^{-1}$ embeds ${\mc A}'$ in $(\Rel(X),*,\sqsubseteq,\emptyset)$.
\end{proof}

We could instead express Theorem \ref{two} in terms of the converse of the refinement order, leading to an alternative way of viewing ordered semigroups with zero in terms of demonic composition and the converse of demonic refinement.

\subsection{Introducing an identity element}  \label{orderzero}

In \cite{hirsch}, it was shown that $R(\comp,\subseteq,1')$ is not finitely axiomatisable, unlike $R(\comp,\subseteq)$.  Here we consider the result of adding $1'$ to the other signatures.  Note it is shown in \cite[Theorem 10]{submitted}  that $R(*,\sqsubseteq,1')$ has no finite axiomatisation.  Here we tie these facts together, and add in the case of $L_0(\comp,\subseteq,1')$.

\begin{pro}  \label{with1}
$R(*,\sqsubseteq,1')=L_0(\comp,\subseteq, 1')=L(\comp,\subseteq,1')\subsetneq R(\comp,\subseteq,1')$, and none has a finite axiomatisation.
\end{pro}
\begin{proof}
From Proposition \ref{embeddings}, $\psi_3$ embeds any subalgebra of $\Rel(*,\sqsubseteq,1')$ into $\Ltrel_0(\comp,\subseteq,1')$, and so 
$$R(*,\sqsubseteq,1')\subseteq L_0(\comp,\subseteq, 1')\subseteq L(\comp,\subseteq,1')=L(*,\sqsubseteq,1')\subseteq R(*,\sqsubseteq,1'),$$
and the equalities follow.  In \cite[Theorems 7, 10]{submitted}, it is shown that $L(\comp,\subseteq,1')$, and indeed $R(*,\sqsubseteq,1')$ is not finitely axiomatisable, and noted that the former class is properly contained in $R(\comp,\subseteq,1')$ (since  the law $s\subseteq 1' \Rightarrow s=1'$ holds in the former class but not the latter).  The class $R(\comp,\subseteq,1')$ also fails to be finitely axiomatisable as shown in \cite{hirsch}.
\end{proof}

\subsection{Negative results involving union}  \label{union}

It is known that $R(\comp,\cup)$ has no finite axiomatisation; see \cite[Theorem~3.1]{andreka}.  So it is natural to ask whether the same holds for $R(*,\ccup)$, and similarly for some of the other cases.  In fact we have the following.

\begin{pro}  \label{ltequal}
$R(\comp,\cup)=L(\comp,\cup)=L_0(\comp,\cup)=T(\comp,\cup)=R(*,\ccup)$.  Hence none can be finitely axiomatised.
\end{pro}
\begin{proof}
Every algebra of binary relations $(A,\comp,\cup)\subseteq (\Rel(X),\comp,\cup)$ embeds in $(\Ltrel_0(X),\comp,\cup)\in L_0(\comp,\cup)$ by Proposition \ref{embeddings}, so $R(\comp,\cup)\subseteq L_0(\comp,\cup)\subseteq L(\comp,\cup)\subseteq R(\comp,\cup)$, so all are equal.  By the same result, ${\mc A}=(A,*,\ccup)\subseteq (\Rel(X),*,\ccup)\in R(*,\ccup)$ embeds in $(\Trel(X),\comp,\cup),$ and so $R(*,\ccup)\subseteq T(\comp,\cup)\subseteq L(\comp,\cup)=L(*,\ccup)\subseteq R(*,\ccup)$.  By Theorem $3.1$ in \cite{andreka}, $R(\comp,\cup)$ cannot be finitely axiomatised.
\end{proof}

By Theorem $3.1$ in \cite{andreka}, $R(\comp,\cup,\emptyset)$ has no finite axiomatisation.  It follows easily from Proposition \ref{embeddings} that $R(\comp,\cup,\emptyset)\subseteq L_0(\comp,\cup,{\bf 0})$, and indeed the inclusion is proper because the law $x\cup \emptyset=x$ holds in the former, but $x\cup {\bf 0}=x$ fails in the latter.

Regarding signatures containing the diagonal relation $1'$, we have the following, some of which is shown in \cite{submitted}.

\begin{pro}  \label{cupwith1}
$R(*,\ccup,1')=L_0(\comp,\cup, 1')=L(\comp,\cup,1')\subsetneq R(\comp,\cup,1')$, and none has a finite axiomatisation.
\end{pro}
\begin{proof}
Proof that the three equalities hold is exactly as for Proposition \ref{with1} but makes use of Proposition \ref{embeddings2} rather than Proposition \ref{embeddings}.  That $R(\comp,\cup,1')$ has no finite axiomatisation follows from Theorem $3.1$ in \cite{andreka}. That $R(*,\ccup,1')$ has no finite axiomatisation follows from Theorem $10$ in \cite{submitted}.  That $L(\comp,\cup,1')\subsetneq R(\comp,\cup,1')$ follows for the same reason that $L(\comp,\subseteq,1')\subsetneq R(\comp,\sqsubseteq,1')$ given earlier: the law $s\subseteq 1' \Rightarrow s=1'$ holds in the former class but not the latter (and recalling that $s\subseteq t$ is equivalent to the equation $\cup t=t$).
\end{proof}

In this subsection, all signatures are purely algebraic, so we can also consider varieties generated by our isomorphism classes.
First note that every member of both $R(\comp,\cup)$ and $R(*,\ccup)$ is an idempotent semiring, meaning a structure $(A,\cdot,+)$ such that 
\bi
\item $(A,\cdot)$ is a semigroup;
\item $(A,+)$ is a semilattice; and 
\item $\cdot$ is additive with respect to $+$, meaning it satisfies the distributive law
$$(a+b)\cdot(c+d)=a\cdot c+a\cdot d+b\cdot c+b\cdot d.$$
\ei

For a class $M$ of algebraic structures, let $VM$ denote the variety generated by $M$.  That $VR(\comp,\cup)$ is the class of idempotent semirings follows from \cite[Proposition~4.4]{andreka}.  But because $R(\comp,\cup)=R(*,\ccup)$, it follows that $VR(\comp,\cup)= VR(*,\ccup)$, and so the latter is also the class of idempotent semirings.

An idempotent semiring has a {\em zero} $0$ if $a+0=a$ for all $a$.  Again by \cite[Proposition~4.4]{andreka}, $VR(\comp,\cup,\emptyset)$ is the class of idempotent semirings with zero.  Then $VR(\comp,\cup,\emptyset)\subsetneq VL_0(\comp,\cup,{\bf 0})$ by the first part of Proposition \ref{embeddings} as well as the fact that $s\cup {\bf 0}=s$ does not hold in the latter.

An idempotent semiring has an {\em identity} $1$ if $a1=1a=a$ for all $a$.  Again by \cite[Proposition~4.4]{andreka}, $VR(\comp,\cup,1')$ is the class of idempotent semirings with identity.  The same result also gives that $VR(\comp,\cup,\emptyset,1')$ is the class of idempotent semirings with zero and identity.

\section{Results for ``mixed" signatures}  \label{mixed}

\subsection{Negative result: $R(*,\subseteq)$ is not finitely axiomatisable}\label{neg}  

In this section we consider the signature obtained by mixing 
demonic composition~$*$ and ``angelic order"~$\subseteq$, 
and consider the representation class 
$R(\subseteq, *)$ consisting of the %=\sub\{(\Rel(X),\subseteq,*): X\text{ a set}\}$
isomorphic copies of substructures of full families of binary relations with demonic composition and inclusion.
Occasionally, we will expand this signature to include the identity constant $1'$ as well.

Valid formulas include the reflexivity, transitivity and antisymmetry of $\subseteq$ and the associativity of $*$, as we saw.   
A further valid property involving both $\subseteq$ and $*$ is right monotonicity: 
$t_1\subseteq t_2\rightarrow s*t_1\subseteq s*t_2$.  
Left monotonicity can fail though, since demonic composition is defined as the conjunction of two parts 
--- the positively monotonic first part of ordinary relation composition, 
while the universally quantified second part is negatively monotonic in $s$, 
the conjunction need not be positively or negatively monotonic.
However, a restricted version of left monotonicity holds:
\begin{equation}\label{eq-val}
(s_0\subseteq s_1\land s_0\subseteq s_1*t)\to s_0*t\subseteq s_1*t
\end{equation}
as can be seen by the following argument.
Assume the antecedant and assume $(x,y)\in s_0*t$. Then there is $z$ such that $(x,z)\in s_0\subseteq s_1$ and $(z,y)\in t$.
Thus $(x,y)\in s_1\comp t$.
We still have to check the second part in the definition of demonic composition.
Since $(x,z)\in s_0\subseteq s_1*t$, 
%we have $u$ such that $(x,u)\in s_1$ and $(u,z)\in t$ and
for every $w$ with $(x,w)\in s_1$, we have $v$ such that $(w,v)\in t$.
Thus the  assumption $(x,w)\in s_1$ implies $(w,v)\in t$ as desired.
The quasi-equation \eqref{eq-val} may be generalised as follows.

\begin{lemma}\label{lem:q}
The quasi-equation
\begin{equation}\label{eq-valn}
(s_0\subseteq s_n\land \bigwedge_{i <n} s_i\subseteq s_{i+1}*t)\to s_0*t\subseteq s_n*t
\end{equation}
is valid over $R(\subseteq, *)$, 
for each $n<\omega$.
\end{lemma}

\begin{proof} 
Assume $s_0\subseteq s_n$ and $s_i\subseteq s_{i+1}*t$ for $i<n$.
We use the fact that $a*b\subseteq a;b$ repeatedly, and that~$;$ is both left and right monotonic.
%over $R(\subseteq, *, ;)$.
So, 
\begin{align}
\label{a:1}s_0*t&\subseteq s_0;t\subseteq s_n;t \\ %\mbox{ and }\\
\intertext{and}
\label{a:2}s_0*t&\subseteq s_0;t\subseteq (s_1*t);t\subseteq s_1;t;t\subseteq\ldots\subseteq s_{n-1};\overbrace{t;\ldots;t}^n\subseteq (s_n*t);\overbrace{t;\ldots;t}^n.
\end{align}
By \eqref{a:1}, if $(x, y)\in s_0*t$, then $(x, y)\in s_n;t$  and by~\eqref{a:2} there is $w$ such that  $(x, w)\in s_n*t$ and $(w, y)\in\overbrace{t;\ldots;t}^n$.  
By $(x, w)\in s_n*t$, we get  $\forall v\;((x, v)\in s_n\rightarrow\exists u\; (v, u)\in t)$.   
It follows, by~\eqref{def:demonic}, that $(x, y)\in s_n*t$.  
This proves $s_0*t\subseteq s_n*t$.
\end{proof}

We will use  quasi-equations \eqref{eq-valn} to design a sequence of non-representable 
structures $(\c A_n: n<\omega)$ 
such that a non-principal ultraproduct $\c A$ is representable.

\subsection*{Non-representable structures}

We define structures in the  signature $(\leq, *,1')$ 
(we write $\leq$ in place of $\subseteq$ to remind us 
that these structures are abstract and might not be representable).  
Fix $n<\omega$.   Let $\Sigma=\set{t, s_i:i\leq n}$ be an alphabet with $n+2$ characters
and $\Sigma^*$ denote the set of finite words, including the empty word $\Lambda$, over $\Sigma$.  $\Sigma^+=\Sigma^*\setminus\set{\Lambda}$.  
We will denote the general elements of $\Sigma^*$ by lower case Greek letters.
% and let $\Lambda$ denote the empty word.
The concatenation of $k$ copies of a character $x\in\Sigma$ is denoted by $x^k$,
so $x^0=\Lambda$.
Define  three binary relations $L, <$ and $\leq$ over $\Sigma^*$ by
\begin{align}
\nonumber
L&=\set{(s_0, s_n)}\cup\set{(s_i, s_{i+1}t):i<n}\\
\label{eq-order}  <&=\set{(\alpha \sigma, \alpha\sigma^+):\alpha\in\Sigma^*,\; (\sigma, \sigma^+)\in L}\\
\nonumber\leq&={} < {} \cup {} \set{(\alpha, \alpha):\alpha\in\Sigma^*}
\end{align} 
and let $\c A_n=(\Sigma^*,\leq, *,\Lambda)$ where $*$ is concatenation 
(thus, for the remainder of this section, we may drop $* $ from $\alpha*\beta$ 
and simply write $\alpha\beta$).  Observe that there are no chains $a<b<c$ in $\c A_n$ (so $<$ and $\leq$ are transitive) and for each $a$ there is at most one $b$ such that $b<a$.
Also note
\begin{equation}\label{eq:alpha beta} \alpha\beta\geq\alpha\rightarrow\beta=\Lambda.
\end{equation}

\begin{lemma}\label{lem-nr}
For $1\leq n<\omega$, the $(\leq,*)$-reduct $\c B_n$ of 
$\c A_n$ is not $(\subseteq,*)$-representable.
\end{lemma}

\begin{proof}
If $\c B_n$ were representable, then it would satisfy~\eqref{eq-valn}.  
But, letting $\alpha$ be the empty string in~\eqref{eq-order}, 
we have $s_0\leq s_n$ and $s_i\leq s_{i+1}t$ for all $i<n$, so the premise of~\eqref{eq-valn} is true.  
However, the conclusion $s_0t\leq s_nt$ fails.
\end{proof}

It remains to prove that a non-principal ultraproduct of the $\c A_n$s is representable.  
In order to do that we define 
%a $k$-round 
a representation game $\Gamma_k(\c A)$.

\subsection*{Representation games}
For a structure $\c A=(A,\le, *, 1')$ and $S, T\subseteq A$, we write $S* T$ for 
$\set{s* t:s\in S,\; t\in T}$.   
%Here we write $\leq$ instead of $\subseteq$ to remind us that $\c A$ is an abstract structure, not necessarily representable. 
We let $S^\uparrow=\set{s^+\in A:(\exists s\in S)s\leq s^+}$, 
and for $a\in A$ we write $a^\uparrow$ for $\set{a}^\uparrow$.
%Similarly, $S^\downarrow=\set{s^-\in A:(\exists s\in S)s\geq s^-}$.  

A \emph{network} $\c N=(X, N, \O)$ over $\c A$ consists of a set $X$ of nodes, 
an edge labelling function $N\colon X\times X\rightarrow \wp(A)$ 
and a function $\O\colon X\rightarrow \wp(A)$ such that $N(x, y)= N(x, y)^\uparrow$
and $1' \in N(x,y)$ iff $x=y$ for all $x, y\in X$.
A network $\c N$ is \emph{consistent} if $N(x, y)\cap \O(x)= \emptyset$ for all $x, y\in X$.  

For two networks $\c N=(X, N, \O)$ and $\c N'=(X', N', \O')$ we write $\c N\subseteq\c N'$ 
if $X\subseteq X'$, $N(x, y)\subseteq N'(x, y)$ and $\O(x)\subseteq \O'(x)$, for all $x, y\in X$.    
%We say that $\c N'$ is a \emph{neat extension} of $\c N$ if 
%$\c N\subseteq\c N'$ and for all $x, y\in X$, we have $N'(x, y)=N(x, y)$ and $\O'(x)=\O(x)$.  
To simplify notation we will use the same symbol $N$ to stand for a network, 
its set of nodes and its edge labelling function.
%and  write $\nodes(N)$  for its set $X$ of nodes.

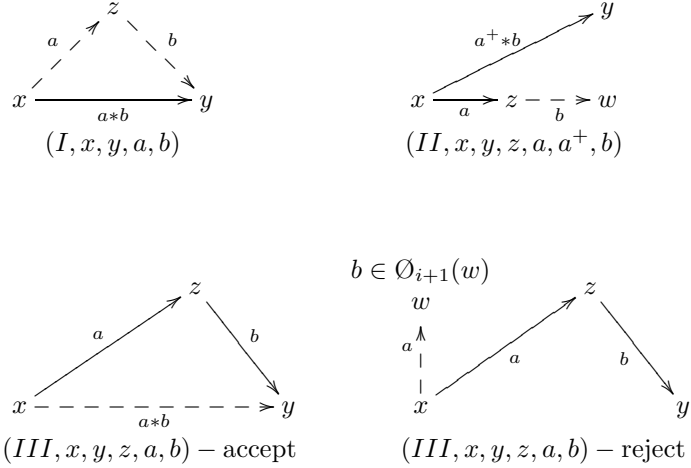
\begin{figure}[H]
\[
\begin{array}{cc}
%\mbox{move}&&\mbox{move}\\
%\hline
\xymatrix{
&z\ar@{-->}[rd]^{b}&\\
x\ar@{-->}[ru]^a\ar@{->}[rr]_{a* b}&&y
}
&
\xymatrix{
&&y\\
x\ar@{->}[rru]^{a^+* b}\ar@{->}[r]_a&z\ar@{-->}[r]_b&w
}
\\
(\ref{en:witness}, x, y, a, b)&(\ref{en:demonic}, x, y, z, a, a^+, b )\\
\xymatrix{&&\\
{\begin{array}{c} \hspace{.8in} \\ \; \end{array}}
&z\ar@{->}[rd]^b&\\
   x\ar@{->}[ru]^a\ar@{-->}[rr]_{a* b}&&y
}
&
\xymatrix{
&&\\
{\begin{array}{c} b\in\O_{i+1}(w)\\ w\end{array}}&z\ar@{->}[rd]_b&\\
x\ar@{->}[ru]_a\ar@{-->}[u]^a&&y
}
\\
\;\;\;\;\;\;\;\;\;\;(\ref{en:choose}, x, y, z, a, b) - \mbox{accept}&\;\;\;\;\;\;\;\;(\ref{en:choose}, x, y, z, a, b) -\mbox{reject}
\end{array}\]
\caption{\label{fig:moves} Three types of $\forall$-moves and  $\exists$'s responses.  
Solid arrows are selected by $\forall$ and dashed arrows are newly labelled or refined by $\exists$.}
\end{figure}

The game $\Gamma_k(\c A)$ is played by two players $\forall$ (universal) and $\exists$ (existential), 
who play a sequence of networks $N_0\subseteq N_1\subseteq\ldots\subseteq N_i\subseteq\ldots$
($i<k$), as follows.  
In the initial round, $\forall$ picks any $a_0\not\leq b_0\in\c A$ and 
$\exists$ has to play a network $N_0$ such that there are two nodes $x_0, y_0$ 
with $a_0\in N_0(x_0, y_0)$ but $b_0\not\in N_0(x_0, y_0)$.    
In a later round if the current network is $N_i$, $\forall$  can make three types of moves.
\begin{enumerate} 
\renewcommand{\theenumi}{\Roman{enumi}} 
\renewcommand{\labelenumi}{\Roman{enumi}} 
\item\label{en:witness}
\emph{Witness move} $(\ref{en:witness}, x, y, a, b)$:
$\forall$ picks $x, y\in N_i$ and  $a,  b\in A$ such that $a* b\in N_i(x, y)$.
$\exists$  is required to respond with $N_{i+1}$ such that $a\in  N_{i+1}(x, z)$ and $ b\in  N_{i+1}(z, y)$ 
for some $z\in N_{i+1}$.
\item\label{en:demonic}  
\emph{Demonic move} $(\ref{en:demonic}, x, y, z, a, a^+, b )$:
$\forall$ picks $x, y, z\in N_i$ and $a,  a^+, b\in A$ where $a\leq a^+$, $a^+* b\in N_i(x, y)$ and 
$a\in N_i(x, z)$.
$\exists$ must play $N_{i+1}$ where there is $w\in N_{i+1}$ with $b\in N_{i+1}(z, w)$.
\item \label{en:choose} 
\emph{Choice move} $(\ref{en:choose}, x, y, z, a, b)$:
$\forall$ picks $x, y, z\in N_i$ and  $a\in  N_i(x, z)$, $b\in N_i(z, y)$.
In this case, $\exists$ may either \emph{accept} and ensure $a* b\in  N_{i+1}(x, y)$  
(by refining the label if necessary),
or \emph{reject} and ensure there is $w\in N_{i+1}$ such that 
$a\in N_{i+1}(x, w)$ and $b\in \O_{i+1}(w)$.
\end{enumerate}
See Figure~\ref{fig:moves} for an illustration of these three kinds of moves. 
A move is \emph{trivial} if $N_{i+1}=N_i$ is a legal response, 
e.g.\  a move $(\ref{en:witness}, x, y, a, 1')$ is trivial, for $a*1'=a\in N(x, y)$.  

In any round, if $\exists$ plays a network $N$ where  $b_0\in N(x_0, y_0)$  or 
$N$ is inconsistent (i.e., there are $x, y\in N$ with $N(x, y)\cap\O(x)\neq\emptyset$), 
then $\forall$ wins.
If $\forall$ does not win in any round, then $\exists$ wins the game.

The game may easily be modified to test $(\subseteq, *)$-representability, 
by simply omitting the requirement $1'\in N(x, y)$ iff $x=y$ from the definition of network.

\begin{pro}\label{pro-game}
Let $\c A$ be countable or finite. 
$\c A$ is $(\subseteq, *,1')$-representable if and only if $\exists$ has a winning strategy in 
$\Gamma_\omega(\c A)$.
\end{pro}
\begin{proof}
See   \cite[theorem 7.19 ]{HH:book} for a similar proof.  
If $\c A$ is representable then  a \ws\ for $\exists$ is to  maintain an embedding from the current network 
into a fixed representation.  
Conversely, if $\c A$ is countable and $\exists$ has a \ws, 
we consider a play in which $\forall$  plays $a\not\leq b$ initially, 
then schedules all possible moves (this is possible by countability) and $\exists$ uses her \ws.    
Let $N(a, b)$ be the  limit of such a play.  
Now let $N$ be the disjoint union of all limit networks $N(a, b)$ where $a\not\leq b$.  
Define $\theta\colon \c A\rightarrow \wp(N\times N)$ by $\theta(a)=\set{(x, y): N(x, y)\leq a}$.  
This map is injective (since we included all $N(a, b)$) and clearly respects $\leq$.  
Since all possible witness moves are played we have $\theta(a* b)\subseteq\theta(a);\theta(b)$.
Furthermore, since all demonic moves are played,  
if $(x, y)\in\theta(a* b)$ then for all $w$ with $(x, w)\in\theta(a)$ we have $w\in dom(\theta(b))$, 
so $\theta(a* b)\subseteq\theta(a)*\theta(b)$.
Conversely, since all choice moves are played and $\exists$ wins, $\theta(a)*\theta(b)\subseteq\theta(a* b)$.  
Hence $\theta$ is a representation.
\end{proof}

%Thus representation games are closely related to representability of abstract structures. 
We have seen that $\c A_n$ is not representable.
%one may simply observe that quasi-equation \eqref{eq-valn} fails in $\c A_n$ (let $\alpha$ be the empty string in \eqref{eq-order}, observe that $s_0t\not\leq s_nt$).  But \nb{Maybe skip this proof and prove \eqref{eq-order} instead} 
To get a feel for the representation games $\Gamma_k(\c A)$ we show that the non-representability of
$\c A_n$ is witnessed by a game.

\begin{lemma}\label{lem-lose}
The universal player $\forall$ has a \ws\ in the representation game $\Gamma_{n+2}(\c A_n)$
of length $n+2$.
\end{lemma}

\begin{proof}
In the initial round, $\forall$ picks   $s_0t\not\leq s_nt$ and $\exists$ responds with a network $N_0$ 
with nodes $\set{0, 2}$ and $N_0(0, 2)=(s_0t)^\uparrow=\set{s_0t}$.  
Then $\forall$ plays $(\ref{en:witness}, 0, 2, s_0, t)$, forcing $\exists$ to include a node $1\in N_1$ 
where $N_1(0, 1)=s_0^\uparrow=\set{s_0, s_n, s_1t}$ and $N_1(1, 2)=t^\uparrow=\set{t}$. 
Then he plays $(\ref{en:choose}, 0, 1, 2, s_n, t)$.  
If $\exists$ accepts, she loses (since we get $s_nt\in N_2(0, 2)$), 
so assume she rejects by adding a node $0^*$ to $N_2$ 
where $N_2(0, 0^*)=s_n^\uparrow$ and $t\in\O_2(0^*)$.  
Then he plays $(\ref{en:witness}, 0, 1, s_1, t)$, forcing $y_1\in N_3$ 
where $N_3(0, y_1)=s_1^\uparrow=\set{s_1, s_2t}$ and $N_3(y_1, 1)=t^\uparrow$.  
Similarly we get $N_{4}(0, y_{2})=s_{2}^\uparrow=\set{s_{2}, s_3t}$.  
Eventually, in $N_{n+1}$, we arrive at $N_{n+1}(0, y_{n-1})=s_{n-1}^\uparrow=\set{s_{n-1}, s_nt}$.  
But then $\forall$ plays $(\ref{en:demonic}, 0, y_{n-1}, w, s_n, s_n, t)$, 
forcing a new node $v\in N_{n+2}$ where $t\in N_{n+2}(w, v)$, so $\exists$ loses.
See Figure~\ref{fig-nr}.
\end{proof}
\begin{figure}[h]
\[
\xymatrix{
v\\
\stackrel{\;(t\in\O(w))}{0^*}\ar[u]^{t} & y_{n-1} & \dots & y_2\ar[rr]^{t} & & y_1\ar[dd]^{t} & \\
&&&&&&\\
0\ar[uu]^{s_n}\ar[uur]^(.6){s_nt}_(.6){s_{n-1}}\ar[uurrr]^(.65){s_3t}_(.65){s_2}\ar[uurrrrr]^(.7){s_2t}_(.7){s_1}
\ar[rrrrr]^{s_1t,s_n}_{s_0}\ar@/_3pc/[rrrrrr]_{s_0t}&&&&&1\ar[r]^{t}&2
}
\]
\caption{Non-representability of $\c{A}_n$.} 
\label{fig-nr}
\end{figure}
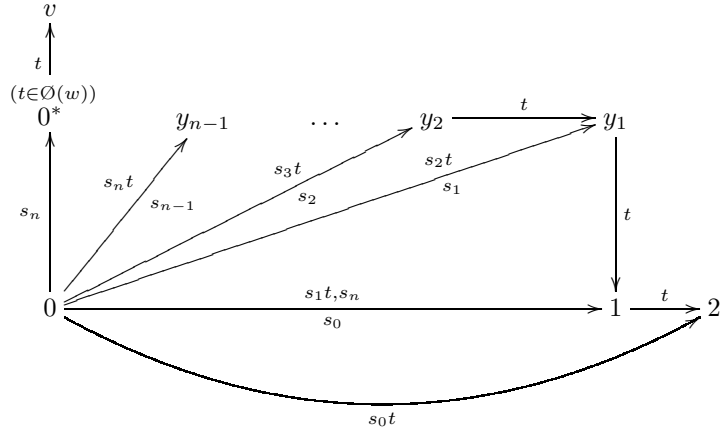

\subsection*{Representing the ultraproduct}

\begin{lemma}\label{lem-surv}
The existential player $\exists$ has a \ws\ in the representation game $\Gamma_{n-1}(\c A_n)$ of length $n-1$. 
\end{lemma}

\begin{proof}
We describe a winning strategy for $\exists$.
In her replies to $\forall$'s moves, she will construct a chain of  networks $N_0\subseteq N_1\subseteq\ldots$,  the labels $N(x, y)$ and $\O(x)$,  once defined,  are never changed in later rounds.   We define these networks by specifying a \emph{grid} $(D,T, f)$ consisting of a set $D$ (of nodes), a set $T\subseteq D$ (of terminal nodes) and a partial reflexive function $f$  from $D\times D$ to $\Sigma^*$ such that
\begin{enumerate}
\renewcommand{\theenumi}{\roman{enumi}}

\item\label{en:lambda} $f(x, y)=\Lambda$ iff $x=y$.
\item \label{h:xyz} if $(x, y), (y, z), (x, z)\in dom(f)$ then $f(x, y)f(y, z)\geq f(x, z)$, 
\item\label{h:ab}  if $f(x, y)=\alpha\beta$ then there is $z$ such that $f(x, z)=\alpha$ and $f(z, y)=\beta$, 
%\item if $f(x, y)=\alpha\beta$ and $f(x, z)=\alpha$ then $(z, y)\in dom(f)$ and $f(z, y)=\beta$.
\item if $(x, y)\in dom(f)$ and $x\in T$ then $y=x$.\label{en:T}
\end{enumerate}
A grid $(D, T, f)$  determines a network $N_{D,T, f}$ with nodes $D$ where 
\begin{align*}
\O_{D,T, f}(w)&=\left\{\begin{array}{ll} \Sigma^+&w\in T\\ \emptyset&w\in D\setminus T\end{array}\right.\\
N_{D,T, f}(x, y)&=\left\{\begin{array}{ll} (f(x, y))^\uparrow&(x, y)\in dom(f)\\
\emptyset&(x, y)\in(D\times D)\setminus dom(f)\end{array}\right.
\end{align*}    

A grid $(D', T', f')$ extends a grid $(D, T, f)$ if (i) $D\subseteq D'$, (ii) $T=T'\cap D$ and (iii) $f=f'\cap(D\times D)$.  [Figure~\ref{fig-nr} with $v$ deleted illustrates a fairly typical grid where  $T=\set{0^*}$.]
Our induction hypotheses for $N_k$:
\begin{enumerate}

\item\label{h:grid} There is a grid  $(D_k,T_k, f_k)$  satisfying \eqref{en:lambda}--\eqref{en:T},  where $N_k=N_{D_k,T_k, f_k}$.

\item\label{h:trans}  If $(x, y), (y, z)\in dom(f_k)$ but $(x, z)\not\in dom(f_k)$ then $f_k(x, y)=t$.

\item\label{h:no beta}  If $f_k(x, z)\leq f_k(x, w)$ and  there is $y$ such that $(x, y), (w, y)\in dom(f)$ and $f_k(w, y) \neq\Lambda$, then $z=w$
\item\label{h:triangle}
If $f_k(x, y)=f_k(x, z)=s_i$ then $y=z$.  
If $i<j,\; f_k(x, y)=s_i$  and $f_k(x, z)=s_j$ then either $j-i\leq k$ and there is  $w$ with  $f_k(x, w)=s_{i+1},\; f_k(w, y)=t$, \/  or $j=n,\;z\in T_k$ and there is $v$ where $f_k(x, v)=s_0$.

\end{enumerate}  

In the initial round, if $\forall$ plays $a\not\leq b$ where $a$ is the string of characters $a_0a_1\ldots a_{p-1}$ ($p=|a|$), then $\exists$ defines a grid $(D_0, \emptyset, f_0)$ where  $D_0= \set{0, 1, \ldots, p}$, there are no terminal nodes,   and  $f_0(i, j)=a[i, j]=a_ia_{i+1}\ldots a_{j-1}$ 
for $i\leq j\leq p$ and the label $f_0(j, i)$ for $j>i$ is undefined.  She plays $N_0=N_{D_0, \emptyset, f_0}$, the network determined by this grid.  Noting that $dom(f_0)$ is a total linear order, it is easy to check all hypotheses.  

In a later round $k+1$,  the current network is $N_k=N_{D_k,T_k, f_k}$, we consider $\forall$'s next move.     If  $N_{k+1}=N_k$ is a legal $\exists$-response to $\forall$'s move we say the move is trivial.
We assume the hypotheses for $N_k$ and  assume that $\forall$ plays only non-trivial moves.    

For a demonic move $(\ref{en:demonic}, x, y, z, \alpha, \alpha^+, \beta)$, we have $\alpha\geq f_k(x, z), \;\alpha^+\beta\geq f_k(x, y), \;\alpha\leq\alpha^+$ and we may assume $\beta\neq\Lambda$ (else the move is trivial), so $f_k(x, y)=\alpha^+\beta^-$ for some $\beta^-\leq\beta$.  By IH~\ref{h:grid} and \eqref{h:ab} there is $w$ where $f_k(x, w)=\alpha^+$ and $f_k(w, y)=\beta^-$.   By   IH~\ref{h:no beta}, $w=z$  but then the move is trivial.  Hence, there are no non-trivial demonic moves.

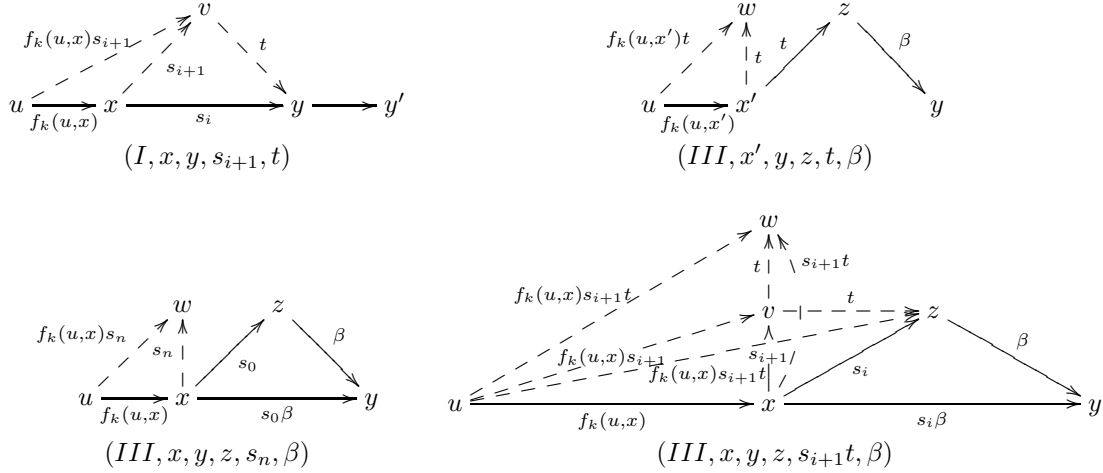
\begin{figure}[h]
\[\begin{array}{cc}
\xymatrix{&&v\ar@{-->}[rd]^t&&\\
u\ar[r]_{f_k(u, x)}\ar@{-->}[rru]^{{f_k(u, x)} s_{i+1}}&x\ar[rr]_{s_i}\ar@{-->}[ru]_{s_{i+1}}&&y\ar[r]&y'
}
&

\xymatrix{
&w&z\ar[rd]^{\beta}&\\
u\ar[r]_{{f_k(u, x')}}\ar@{-->}[ru]^{{f_k(u, x')} t}&x'\ar[ru]^t\ar@{-->}[u]_t&&y
}
\\
(\ref{en:witness}, x, y, s_{i+1}, t)

& (\ref{en:choose}, x', y, z, t, \beta)\\ \\
 
\xymatrix{
&&&\\&w&z\ar[rd]^{\beta}&\\
u\ar[r]_{{f_k(u, x)}}\ar@{-->}[ru]^{{f_k(u, x)} s_n}&x\ar[rr]_{s_0\beta}\ar[ru]_{s_0}\ar@{-->}[u]^{s_n}&&y
}

&
\xymatrix@C=5em{&&w&&\\
&&v\ar@{-->}[u]^t\ar@{-->}[r]^t&z\ar[rd]^{\beta}&\\
u\ar[rr]_{f_k(u, x)}\ar@{-->}[rru]|{ {f_k(u, x)} s_{i+1}} \ar@{-->}[rruu]^{{f_k(u, x)} s_{i+1}t}   \ar@{-->}[rrru]_{{f_k(u, x)} s_{i+1}t}   &&x\ar[rr]_{s_i\beta}\ar[ru]_{s_i}\ar@{-->}[u]|{s_{i+1}}\ar@{-->}@/_1pc/[uu]_>>>>>{s_{i+1}t}&&y
}
\\
  (\ref{en:choose}, x, y, z, s_n, \beta)&{(\ref{en:choose}, x, y, z, s_{i+1}t, \beta)}\end{array}
\]
\caption{\label{fig:strat} $\exists$'s strategy for a witness move and  three rejected choice moves.     $u$ is a typical node in  $D_k$.     $w$ is a  new  terminal node.   $v$ is a new non-terminal node, except in the fourth diagram $v$ is not new if a suitable node is already in $N_k$.  In the second case $f_k(x, y)$ is undefined. Reflexive edges are labelled $\Lambda$ (not shown).  Solid arrows denote edges in $dom(f_k)$,  dashed arrows denote new edges in $dom(f_{k+1})\setminus dom(f_k)$ added by $\exists$ in the current round. $f_{k+1}$ is undefined at all other edges incident with  a new node.    }
\end{figure}

The remaining cases are witness moves and choice moves, and here she will define her response by defining a grid $(D_{k+1},T_{k+1}, f_{k+1})$  extending  the grid $(D_k,T_k, f_k)$ and playing $N_{k+1}=N_{D_{k+1}, T_{k+1}, f_{k+1}}$.

Consider a witness move $(\ref{en:witness}, x', y, \alpha, \beta)$ where  $\alpha\beta\geq f_k(x', y)$.  It is trivial if  either $\beta=\Lambda$ or $\alpha\beta=f_k(x', y)$ by  IH~\ref{h:grid} and \eqref{h:ab}, so suppose $\beta\neq\Lambda,\; \alpha\beta >f_k(x', y)$.      Then    $f_k(x', y)=\gamma s_i, \; \alpha=\gamma s_{i+1}$ and $\beta=t$  (some $i<n$, some $\gamma$), by \eqref{eq-order} .  
    % By IH~\ref{h:term}, $y$ is non-terminal.
  By IH~\ref{h:grid} and \eqref{h:ab}, there is $x\in N_k$ such that $f_k(x', x)=\gamma, \; f_k(x, y)=s_i$.
 $\exists$ calculates her response by pretending that $\forall$ instead played $(\ref{en:witness}, x, y, s_{i+1}, t)$.  Suppose there is $z$ such that $f_k(x, z)=s_{i+1}$.  By \eqref{h:triangle}, either $f_k(z, y)=t$, so the move $(\ref{en:witness}, x, y, s_{i+1}, t)$ is trivial, or $i+1=n,\; z\in T_k$ and $f_k(x, v)=s_0$ (for some $v$).  In the latter case we have $f_k(x, v)=s_0,\; f_k(x, y)=s_{n-1}$, but this contradicts \eqref{h:triangle} since $k<n-1$.  So there is no node $z\in N_k$ where $f_k(x, z)=s_{i+1}$.  For her response, she adds a single new node $v$ to $N_k$, lets $f_{k+1}(x, v)=s_{i+1},\; f_{k+1}(v, y)=t$ and $f_{k+1}(u, v)=f_k(u, x)s_{i+1}$ whenever $f_k(u, x)$ is defined (and $f_k$ is not defined on other edges incident with $v$), as shown in the first part of  figure~\ref{fig:strat}. There are no new terminal nodes ($T_{k+1}=T_k$) and she plays $N_{k+1}=N_{D_{k+1}, T_{k+1}, f_{k+1}}$.      Note that $f_{k+1}(x', z)=\gamma s_{i+1}$, so the response is legal response to $(\ref{en:witness}, x', y, \gamma s_{i+1}, t)$.   
 
 We briefly check our four induction hypotheses.
 \eqref{h:grid} holds by  definition of $N_{k+1}$.  For \eqref{h:trans}, whenever $f_k(u, u'), f_{k+1}(u', z)$ are defined then so is $f_{k+1}(u, z)$, and $f_{k+1}(z, w)\neq \Lambda \rightarrow (w=y\wedge f_{k+1}(z, w)=t)$.
  Since there was no $v\in N_k$ with $f_k(x, v)=s_{i+1}$, the first part of  \eqref{h:triangle} holds; the second part also holds by \eqref{h:triangle} for $N_k$.  For \eqref{h:no beta}, the only case where $f_{k+1}(u, z)>f_k(u, v)$ could happen is when $i+1=n,\; f_{k+1}(u, z)=\gamma s_n,\; f_k(u, v)=\gamma s_0$, but this cannot happen since $k<n-1$.  Thus all induction hypotheses are maintained.

Now, consider a choice move $(\ref{en:choose}, x', y, z, \alpha, \beta)$ played on $N_k$, where $\alpha\geq f_k(x', z), \;\beta\geq f_k(z, y)$.  First suppose $f_k(x', z)=\alpha$.  By IH~\ref{h:grid} and \eqref{h:xyz} either $f_k(x', z)f_k(z, y)\geq f_k(x', y)$ (so $\alpha\beta\geq f_k(x', y)\in N_k(x', y)$ and the move is trivial) or $(x', y)\not\in dom(f_k),\; f_k(x', z)=t$ (by IH~\ref{h:trans}), in which case she adds a single new terminal node $w\in T_{k+1}$ to the network where $f_{k+1}(x', w)=t$ and more generally if $f_k(u, x')$ is defined then she lets $f_{k+1}(u,w)=f_k(u, x')t$, but $f_{k+1}$ is undefined on other edges incident with $w$, as illustrated in the second part of figure~\ref{fig:strat}.   Since $w\in T_{k+1}$ we have $\beta\in \O_{k+1}(w)=\Sigma^+$, so $N_{k+1}$ is a legal rejection of the choice move.    We always included $(u, w)$ in $dom(f_{k+1})$ whenever $(u, u'), (u', w)\in dom(f_{k+1})$, we never have $f_{k+1}(u, w)=s_i$ (any $i\leq n$) and there are no outgoing irreflexive edges from $w$, so all induction hypotheses are maintained.

  Otherwise, $f_k(x', z)<\alpha$, so $(f_k(x', z), \alpha)=(\gamma s_0, \gamma s_n)$ or $(\gamma s_i, \gamma s_{i+1}t)$ (some $\gamma$, some $i<n$).    As with witness moves, by  IH~\ref{h:grid} and \eqref{h:ab}, there is $x\in N_k$ such that $f_k(x', x)=\gamma$ and $f_k(x, z)$ is $s_0$ or $s_i$ depending on the case.  She pretends that $\forall$ made the move $(\ref{en:choose},x, y, z, s_n, \beta)$ or  $(\ref{en:choose},x, y, z, s_{i+1}t, \beta)$ (respectively) to calculate her response.
  These cases are illustrated  in the third and fourth parts of figure~\ref{fig:strat}.  For the former case,  we have $f_k(x, z)=s_0$ and she  adds a single new terminal node $w\in T_{k+1}$ and lets $f_{k+1}(x, w)=s_n$, more generally when $N_k(u, x)\neq\emptyset$ she lets $f_{k+1}(u, w)=f_k(u, x) s_n$,   but $f_k$ is undefined on other edges incident with $w$.    Observe that $f_{k+1}(x', w)=\gamma s_n$, so this is a legal rejection of $(\ref{en:choose}, x', y, z, \gamma s_n, \beta)$
  
  In the latter case, we have $f_k(x, z)=s_i$ and some care is needed to maintain \eqref{h:triangle}.  There might be a node $v\in N_k$ where $f_k(x, v)=s_{i+1}$.  If so, by \eqref{h:triangle} either $v\in T_k,\; i+1=n$ and $f_k(x, v')=s_0$ (some $v'$), but this is impossible as $k<n-1$,  so in fact the former alternative in \eqref{h:triangle} must  hold, and we have $f_k(v, y)=t$.  If there is no such node  then she adds  a new node $v$ (non-terminal) and lets $f_{k+1}(x, v)=s_{i+1},\; f_{k+1}(v, z)=t,\; f_{k+1}(u, v)=f_k(u, x)s_{i+1}$ whenever $f_k(u, x)$ is defined, this case is illustrated in the last part of figure~\ref{fig:strat}.  Whether the node $v$ was already present or newly added, she goes on to  add a new terminal node $w\in T_{k+1}$ and lets $f_{k+1}(v, w)=t,\;f_{k+1}(x, w)=s_{i+1}t,\; f_{k+1}(u, w)=f_k(u, x)s_{i+1}t$ whenever $f_k(u, x)$ is defined.  $f_{k+1}$ is undefined on other edges incident with $w$ (or incident with $v$ in the case where $v$ is new).  This response (for the case where $v$ is new) is illustrated in the final part of figure~\ref{fig:strat}.
Again, we have $f_{k+1}(x', w)=\gamma s_{i+1}t$, so this is a legal rejection of $(\ref{en:choose}, x', y, z, \gamma s_{i+1}t, \beta)$.   As before, our induction hypotheses were maintained.
 
Since $N_{k+1}$ is determined by a grid and $\Lambda\not\in\O_{k+1}(w)$ (all $w$) it follows, by IH~\ref{h:grid}\eqref{en:T}, that $N_{k+1}$ is a consistent network, hence $\exists$ has not lost.
\end{proof}

\begin{thm}  
Neither of the representation classes $R(\subseteq, *), R(\subseteq, *,1')$ is finitely axiomatisable.
\end{thm}

\begin{proof}  
See \cite[Theorem 10.14]{HH:book} for a proof in a more general setting, here we just sketch an outline proof.  
Let $S=(\subseteq, *)$ or $S=(\subseteq, *,1')$, and write $\c A_n$ instead of 
`the reduct of $\c A_n$ to $S$'.
By lemma~\ref{lem-nr}, as $n$ grows $\exists$ has a \ws\ in longer and longer games for 
$S$-representability $\Gamma_{n-1}(\c{A}_n)$.  
The existence of such a \ws\ may be expressed by a first order $S$-formula $\sigma_{n-1}$.
It follows that she has a \ws\ in $\Gamma_n(\c{A})$ played on a non-principal ultraproduct 
$\c{A}$ of $(\c{A}_i:i<\omega)$, for all $n<\omega$.
Since $\exists$'s strategy may be reduced to strategies involving a choice of at most two moves, 
we may apply K\"onig's tree lemma to deduce that $\exists$ has a \ws\ in $\Gamma_\omega(\c A)$ 
(or prove such a \ws\ directly without assuming a bound on the number of choices for $\exists$, 
using the ultraproduct construction).   
She also has a \ws\  in $\Gamma_{\omega}(\c A')$ where $\c A'$ is a countable elementary subalgebra 
of $\c A$.
By proposition~\ref{pro-game},  $\c{A}'$ is $S$-representable 
(see  \cite[Proposition 5.2]{HH:wollic} for a similar proof).
% Since the class of representable algebras is closed under elementary equivalence, \nb{why?} $\c A$ is also representable.  
%Hence the complement of $R(\le, *)$ is not closed under ultraproducts.  

Suppose for contradiction that $\sigma$ defines $R(S)$, so $\c A'\models\sigma$, 
hence by elementary equivalence $\c A\models\sigma$.  
By \Los' theorem, $\c A_n\models\sigma$ for `many' $n<\omega$, in contradiction to the 
non-$S$-representability of all $\c A_n$s.
\end{proof}

%*** How about adding the empty relation to these signatures?  Still NFB presumably. 

\subsection{A finite axiomatisation for $R(\cdot,\subseteq)$}  \label{pos}

We earlier saw that the so-called constellation product of binary relations arises fairly naturally in the setting of total correctness.  To contrast with the previous negative result for the ``mixed" angelic/demonic signature $(*,\subseteq)$, we now consider the ``mixed" signature $(\cdot,\subseteq)$ obtained by replacing demonic composition by constellation product, and show it has a finite axiomatisation. 

So consider $x,y\in \Rel(X)$.
If $x\cdot y$ exists, then clearly $x\cdot y=x* y=x;y$ holds, 
and indeed constellation product on binary relations satisfies the following laws:
\begin{enumerate}\renewcommand{\theenumi}{\Roman{enumi}}
\item if $x\cdot(y\cdot z)$ exists then so does $(x\cdot y)\cdot z$ and the two are equal,\label{en:E0}
\item if $x\cdot y$ and $y\cdot z$ exist then  $x\cdot(y\cdot z)$ exists.\label{en:E} 
%\nb{the laws cannot be  combined as Robin claimed}
\end{enumerate}
These are the first two laws of constellations as in~\cite{constell}, and indeed if we define domain $D$ on $\Rel(X)$ as before
(i.e. $D(x)$ is the restriction of the diagonal relation to the domain of $x$), 
then as well as satisfying \eqref{en:E0} and \eqref{en:E}, \/  $(\Rel(X), \cdot, D)$ satisfies some further laws involving $D$ and constellation product and  is therefore a constellation.
Hence we call a partial algebra satisfying \eqref{en:E0} and \eqref{en:E} a {\em pre-constellation}.   
It follows that any collection of binary relations closed under the partial operation of constellation product forms a pre-constellation.

%the following two laws: (i) $D(x)$ is the unique left identity for $x$, \/ (ii)  if $x\cdot D(y)$ is defined then it is $x$, \/ in other words $(Rel(X),\cdot,D)$ forms a constellation in the sense of \cite{inductive}.    A constellation is \emph{normal} if $D(x)\cdot D(y)= D(x)$ and $D(y)\cdot D(x)=D(y)$ imply $D(x)=D(y)$ always holds. %\nb{Can't we assume all the domain laws?  E.g. $DD(x)=D(x)$ etc.?} (clearly valid over $\Rel(\cdot, D)$.
%A constellation is functionally  representable\nb{Exactly what is proved in \cite{constell} and what follows?} if and only if it is  relationally representable if and only if it is normal  \cite{constell},.\nb{changed slightly, please check I haven't screwed it up R}
%since a representation as partial functions is a special kind of representation as binary relations.

%In the case of partial functions, the inclusion order is expressible in the language $(\cdot, D)$ of constellations, by $s\subseteq t$ if and only if $s=D(s)\cdot t$.  
%However, for binary relations this is not the case.  Indeed we show elsewhere that $R(\subseteq, \cdot,D)$ is not finitely axiomatisable.  It is therefore of interest that the less expressive signature $R(\subseteq, \cdot)$ {\em is} finitely axiomatisable.
 
We say $(P,\le,\cdot)$ is an {\em ordered pre-constellation} if $(P,\cdot)$ satisfies~\eqref{en:E0} 
and~\eqref{en:E}, 
$(P,\leq)$ is a poset, and for all $s,t,u,v\in P$ for which $s\leq u$ and $t\leq v$, we have 
\begin{enumerate}\renewcommand{\theenumi}{\Roman{enumi}}
\setcounter{enumi}{2}
\item $s\cdot t\leq u\cdot v$ whenever the two products exist, \label{en:prod}
\item  if $u\cdot t$ exists then $s\cdot v$ exists.\label{en:L}
\end{enumerate}
Observe that $(\Rel(X),\subseteq, \cdot)$ is an ordered pre-constellation, whence so is any subset of $\Rel(X)$ that is closed under constellation products where they exist and equipped with inclusion.  Hence our proof of the next result requires us to show the converse of this.

\begin{thm}   \label{orderembed}	
%Every ordered pre-constellation is embeddable in some $(\Rel(X),\subseteq,\cdot)$.  
$R(\subseteq, \cdot)$ is the class of ordered pre-constellations.
\end{thm}
\begin{proof}
Let $(P,\le,\cdot)$ be an ordered pre-constellation, and let $X=P\cup\{e\}$, where $e\not\in P$.  
We define a map $\theta: P\rightarrow \Rel(X)$ %\wp(X\times X)$ 
by letting $p\mapsto \rho_p$ given by
\[\rho_p=\set{(x, y): x\cdot p\mbox{ is defined and } y\leq x\cdot p}\cup\set{(e, y):y\leq p}\]
for $p\in P$.  We shall show that $\theta$ is an embedding. 

Suppose $s\cdot t$ exists.  
If $(x,y)\in \rho_s$ where $x,y\in P$, then $y\leq x\cdot s$ which exists, 
so we infer from our laws the existence of $x\cdot (s\cdot t)=(x\cdot s)\cdot t\geq y\cdot t$
so  $(y, y\cdot t)\in \rho_t$ and $y\in \dom(\rho_t)$.  
If $(e,y)\in \rho_s$, then $y\leq s$, so $y\cdot t\leq s\cdot t$ exists and again $y\in \rho_t$.  
So $\rho_s\cdot \rho_t$ exists.  
Conversely, suppose $\rho_s\cdot \rho_t$ exists.  Since $(e,s)\in \rho_s$, 
it must be that $s\in dom(\rho_t)$, so $s\cdot t$ must exist.  
So $\rho_s\cdot \rho_t$ exists if and only if $s\cdot t$ exists.  

Now suppose $s\cdot t$ exists.  
Suppose $(x,z)\in \rho_{s\cdot t}$, where $x,z\in P$.  
Then $z\leq x\cdot(s\cdot t)=(x\cdot s)\cdot t$, so letting $y=x\cdot s$, 
we have that $(x,y)\in \rho_s$ and $(y,z)\in \rho_t$, so $(x,z)\in \rho_s\cdot \rho_t$. 
If $(e,z)\in \rho_{s\cdot t}$ then $z\leq s\cdot t$, so $(e,s)\in \rho_s$ and $(s,z)\in \rho_t$, 
so $(e,z)\in \rho_s\cdot \rho_t$.  
Hence $\rho_{s\cdot t}\subseteq \rho_s\cdot\rho_t$.  

Conversely, suppose $(x,z)\in\rho_s\cdot \rho_t$, where $x,y\in P$.  
Then there is $y\in P$ such that $(x,y)\in \rho_s$ and $(y,z)\in \rho_t$.  
So $y\leq x\cdot s$ and $z\leq y\cdot t\leq (x\cdot s)\cdot t$, with the latter existing and equalling 
$x\cdot (s\cdot t)$ because $x\cdot s$ and $s\cdot t$ exist.  
So $(x,z)\in \rho_{s\cdot t}$.  Suppose $(e,z)\in \rho_s\cdot \rho_t$.  
Then there is $y\in P$ such that $(e,y)\in \rho_s$ and $(y,z)\in \rho_t$, so $y\leq s$ and 
$z\leq y\cdot t\leq s\cdot t$, so $(e,z)\in \rho_{s\cdot t}$.  
Hence $\rho_s\cdot \rho_t\subseteq \rho_{s\cdot t}$, and so $\rho_s\cdot \rho_t=\rho_{s\cdot t}$.

Suppose $s\leq t$.  
If $(x,y)\in \rho_s$ where $x,y\in P$, then $y\leq x\cdot s\leq x\cdot t$, 
and so $(x,y)\in \rho_t$.  If $(e,y)\in \rho_s$ then $y\leq s\leq t$, so $(e,y)\in \rho_t$.  
Hence $\rho_s\subseteq \rho_t$.  
Conversely, if $\rho_s\subseteq \rho_t$ then because $(e,s)\in \rho_s$, we have $(e,s)\in \rho_t$, 
and so $s\leq t$.  So $s\leq t$ if and only if $\rho_s\subseteq \rho_t$.   
Hence $\theta$ respects inclusion and is injective, showing it is an embedding.
\end{proof}

This proof generalises Zarecki\u{\i}'s original proof for ordered semigroups in \cite{zarecki}.  
That result may be recovered from the above result as the special case in which all products exist, in which case the axioms for ordered semigroups are recovered, and each $\rho_s$ is a left total relation.  
Another corollary may be obtained by letting $\leq$ be equality, i.e. $\rho_s=\set{(x, x\cdot a):x\in P}\cup\set{(e, a)}$: 
in this case each $\rho_s$ is functional and we obtain a proof that every pre-constellation embeds into an algebra of partial transformations over some $X$ (which embeds into $\Rel(X)$).  
By assuming both that $\leq$ is equality, and that all products exist, 
we recover the familiar Cayley representation of a semigroup as transformations on a set.

It is now relatively easy to add $\emptyset$ to the signature.  Let us say $P$ is an {\em ordered pre-constellation with zero} if it is an ordered constellation such that for all $s,t\in P$,
$$0\cdot s=0,\ \exists (s\cdot t)=0\Rightarrow s=0,\ \exists (s\cdot  0)\Rightarrow s=0,\ 0\leq s.$$
It is easy to check that $(\Rel(X),\cdot,\emptyset,\subseteq)$ is an ordered pre-constellation with zero.

Let $P$ be a pre-constellation, with $0\not\in P$.  Define $P^0$ to be the partial algebra on $P\cup\{0\}$ obtained by extending $\leq$ and $\cdot$ to all of $P\cup\{0\}$ as follows: $0\cdot 0=0$, $s\cdot 0=0$ for all $s\in P$, $0\leq 0$, and $0\leq s$ for all $s\in P$.  

The following is easily checked by simple case analyses.

\begin{pro}
If $P$ is a pre-constellation, with $0\not\in P$, then $P^0$ is an ordered pre-constellation with zero $0$.  Moreover every ordered pre-constellation with zero arises in this way.  If $P\in R(\cdot,\subseteq)$ then $P^0\in R(\cdot,\subseteq,\emptyset)$.
\end{pro}

%Now every member of $R(\subseteq, \cdot)$ is an ordered pre-constellation with zero. 
 Showing the converse gives the following.

\begin{cor}
$R(\subseteq, \cdot,\emptyset)$ is the class of ordered pre-constellations with zero.
\end{cor}
\begin{proof}
If $P$ is an ordered pre-constellation with zero, $Q=P\backslash \{0\}$ is an ordered pre-constellation.  Represent it within $\Rel(Q\cup\{e\})$ as in the  proof of theorem~\ref{orderembed}. Extend the representation to $P$ by mapping $0$ to $\emptyset$.  Hence $P\in R(\cdot,\subseteq,\emptyset)$.
% giving the relation algebra $Q'$ isomorphic to $Q$, so by the last result,
 \end{proof}

\section{Open problems} 

We showed in Section~\ref{neg} that the ``mixed" case $R(*,\subseteq)$ has no finite axiomatisation, and examining the proofs, we might wonder whether an axiomatisation consisting of the axioms of partial order, associativity, right monotonicity plus all instances of quasi-equation \eqref{eq-valn} is complete for $R(\subseteq, *)$.  It is also open as to whether the other mixed signature $R(\comp,\sqsubseteq)$ has a finite axiomatisation.  

We showed in Section~\ref{pos} that $R(\cdot,\subseteq)$ is finitely axiomatised as the class of ordered pre-constellations, but the status of the arguably more natural class $R(\cdot,\sqsubseteq)$ is open.  We note that the apparently independent infinite family of laws hold: for all $n>0$,
$$(s\sqsubseteq t)\wedge 
(\exists ((\ldots(x\cdot t)\cdot u_1)\cdot u_2)\ldots)\cdot u_n) 
\Rightarrow 
\exists ((\ldots(x\cdot s)\cdot u_1)\cdot u_2)\ldots)\cdot u_n.$$ 

We do not know whether $L_0(\comp,\cup,{\bf 0})$ is finitely axiomatisable.  Recall that $T(\comp,\subseteq,\nabla)=R(*,\sqsubseteq,\emptyset)$, the class of dual ordered semigroups with zero; replacing order by join, it follows from Proposition \ref{embeddings} that $R(*,\ccup,\emptyset)\subseteq T(\comp,\cup,\nabla)$, but we do not know if equality holds nor whether either is finitely axiomatisable.  It would also be interesting to know whether $R(*,\ccup,\emptyset)$ and $T(\comp,\cup,\nabla)$ are equal and whether either is finitely axiomatisable (and similarly for the varieties they generate).  
Recall that $VR(\comp,\cup,\emptyset)\subsetneq VL_0(\comp,\cup,{\bf 0})$, and that the former class is finitely axiomatisable, but we do not know whether $VL_0(\comp,\cup,{\bf 0})$ is.  Although $L(\comp,\cup,1')\subsetneq R(\comp,\cup,1')$, we do not know whether 
$VL(\comp,\cup,1')\subsetneq VR(\comp,\cup,1')$

%\bibliographystyle{alpha}
%\bibliography{robin}
 
\section*{Acknowledgements}

The authors would like to thank Associate Professor Marcel Jackson for the fruitful discussions he had with the third author that helped give rise to the material in Section 2.

\end{document}